\documentclass[runningheads]{llncs}
\usepackage[T1]{fontenc}
\usepackage{cite}
\usepackage{amsmath,amssymb,amsfonts}

\usepackage{amsthm}
\usepackage{mathrsfs}
\usepackage{graphicx}
\usepackage{textcomp}
\usepackage{xcolor}
\usepackage{colortbl}
\usepackage{multirow}
\usepackage[normalem]{ulem}
\usepackage{booktabs}
\usepackage{setspace}
\usepackage{soul}
\usepackage{color, xcolor}
\graphicspath{{./photo/}}

\usepackage{algorithmic}
\usepackage{algorithm}
\usepackage{caption} % Add the caption package

\begin{document}
\title{Gradient Flow of Energy: A General and Efficient Approach for Entity Alignment Decoding}
\titlerunning{A General and Efficient Approach for Entity Alignment Decoding}
% If the paper title is too long for the running head, you can set
% an abbreviated paper title here
%

\author{Yuanyi Wang\inst{1} \and
Haifeng Sun\inst{1} \and Lei Zhang\inst{2} \and Bo He\inst{1} \and Wei Tang\inst{3} \and Jingyu Wang\inst{1} \and Qi Qi\inst{1} \and Jianxin Liao\inst{1}}
\authorrunning{Yuanyi Wang et al.}
% First names are abbreviated in the running head.
% If there are more than two authors, 'et al.' is used.
%
\institute{State Key Laboratory of Networking and Switching Technology, \\ Beijing University of Posts and Telecommunications, Bejing, China  \and
China Unicom, Beijing, China. \and
Huawei Translation Services Center.
\email{\{wangyuanyi,hfsun,hebo,wangjingyu,qiqi8266\}@bupt.edu.cn,} \\
\email{zhangl83@chinaunicom.cn}\quad
\email{tangwei133@huawei.com}}

% \author{Anonymous Author}
% \institute{Anonymous Institute}

%
\maketitle              % typeset the header of the contribution
\begin{abstract}
Entity alignment (EA), a pivotal process in integrating multi-source Knowledge Graphs (KGs), seeks to identify equivalent entity pairs across these graphs. Most existing approaches regard EA as a graph representation learning task, concentrating on enhancing graph encoders. However, the decoding process in EA - essential for effective operation and alignment accuracy - has received limited attention and remains tailored to specific datasets and model architectures, necessitating both entity and additional explicit relation embeddings. This specificity limits its applicability, particularly in GNN-based models. To address this gap, we introduce a novel, generalized, and efficient decoding approach for EA, relying solely on entity embeddings. Our method optimizes the decoding process by minimizing Dirichlet energy, leading to the gradient flow within the graph, to maximize graph homophily. The discretization of the gradient flow produces a fast and scalable approach, termed \textbf{T}riple \textbf{F}eature \textbf{P}ropagation (TFP). TFP innovatively generalizes adjacency matrices to multi-views matrices: \textit{entity-to-entity}, \textit{entity-to-relation}, \textit{relation-to-entity}, and \textit{relation-to-triple}. The gradient flow through generalized matrices enables TFP to harness the multi-view structural information of KGs. Rigorous experimentation on diverse public datasets demonstrates that our approach significantly enhances various EA methods. Notably, the approach achieves these advancements with less than 6 seconds of additional computational time, establishing a new benchmark in efficiency and adaptability for future EA methods.
% The code and data are available at https://github.com/wyy-code/TFP.
% \footnote[1]{The code is provided in https://anonymous.4open.science/r/TFP-9425.}

\keywords{Entity Alignment  \and Knowledge Graphs  \and Dirichlet Energy  \and Graph Homophily  \and Decoder \and Triple Feature Propagation.}
\end{abstract}
\section{Introduction}
% \raggedbottoms
Entity alignment (EA), essential for integrating Knowledge Graphs (KGs), tackles a key Semantic Web challenge—merging diverse graph data efficiently \cite{wang2022facing, dsouza2023iterative}. This process, aimed at identifying corresponding entities across distinct KGs, typically involves encoding and decoding phases (Fig.\ref{encoder_decoder} (1)). Current EA methods depend on seed alignments to train encoders for generating entity and, optionally, relation embeddings, which are then decoded to establish alignments.

Recently, EA is approached as a graph representation learning task, with emphasis on improving graph encoders, which fall into two main categories: translation-based models, such as TransE \cite{bordes2013translating} and its variants \cite{chen2017multilingual, guo2019learning}, and Graph Neural Networks (GNN)-based models, like GCN-Align \cite{wang2018cross}, MRAEA \cite{mao2020mraea}, and Dual-AMN \cite{mao2021boosting}. The former views relation embeddings as translation vectors between entities, while the latter aggregates neighboring embeddings to generate entity representations. Recent efforts have enhanced encoder training with additional information, yielding superior alignment accuracy in models like TEA-GNN \cite{xu2021time} and DualMatch \cite{liu2023unsupervised}, which utilize temporal data, and UMAEA \cite{chen2023rethinking} and DESAlign \cite{wang2024towards}, which incorporate images and entity names.

\begin{figure}[t]
\centering
\includegraphics[width = \linewidth]{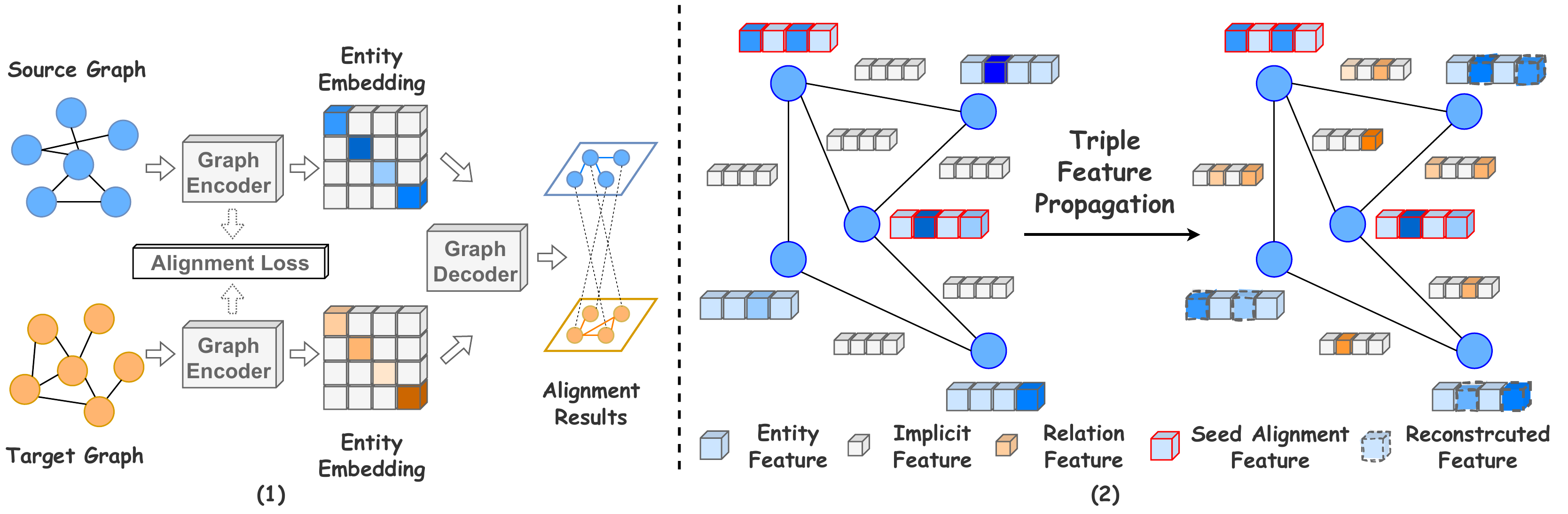}
\caption{(1) The architecture of existing EA methods that involve the encoding and decoding phase. (2) TFP maintains seed alignment features while reconstructing other entity features. It decodes relation features, not directly provided by encoders, to make them explicit based on entity features.} 
\label{encoder_decoder}
\vspace{-0.4cm}
\end{figure}

However, despite advancements in encoders, decoder research \cite{sun2020benchmarking} remains underdeveloped, often constrained by dataset and model specifics. Commonly used, the \textit{greedy search} decoder, is flawed by potentially aligning multiple entities to a single target, contradicting EA's one-to-one mapping principle. Alternative methods like CSLS \cite{lample2018word} and global alignment strategies \cite{xu2020coordinated, zhu2021raga} attempt to address these limitations using algorithms such as the Hungarian \cite{kuhn1955hungarian} or Sinkhorn \cite{cuturi2013sinkhorn}, yet without deeply leveraging KG structural properties. DATTI \cite{mao2022effective} progresses by employing Third-order Tensor Isomorphism for decoding, but its utility is restricted by the necessity for explicit relation embeddings, which many GNN-based encoders do not provide.

In light of these considerations, we propose that an efficient EA decoder should fulfill two key criteria: (\romannumeral1) the ability to exploit the structural information inherent in KGs, and (\romannumeral2) the capacity to generalize across various types of graph encoders. Regarding the first criterion, it is essential to consider both the distribution characteristics of the graph representations generated by the encoders and the intrinsic structural information of the KGs during the decoding phase. For the second criterion, the approach should be adaptable to both translation-based and GNN-based encoder architectures. A notable distinction lies in the fact that translation-based encoders typically produce explicit relation embeddings along with entity embeddings, whereas most GNN-based encoders primarily generate entity embeddings. Despite this, some GNN-based encoders \cite{mao2021boosting} do extract explicit relation embeddings based on entity embedding, underscoring that entity embeddings are central to the graph representation.

In this work, we propose a novel EA decoding approach, Triple Feature Propagation (TFP), that innovatively employs the structural information of KG to reconstruct the entity embeddings by maximizing homophily through minimizing the Dirichlet energy of the entity embedding. Unlike existing methods \cite{fey2019deep, heimann2021refining}, TFP utilizes the gradient flow of Dirichlet energy for efficient information propagation across various graph encoder types, avoiding the need for supplementary information. This approach is a multi-view propagation-based entity embedding reconstruction step, i.e., the decoding process, subsequent to the graph encoding phase. Specifically, TFP generalizes the traditional adjacency matrix to multi-view matrices—\textit{entity-to-entity}, \textit{entity-to-relation}, \textit{relation-to-entity}, and \textit{relation-to-triple}—to depict a comprehensive representation of KG structures. It combines these matrices and minimizes Dirichlet energy to initiate a gradient flow within the graph, facilitating effective information propagation. The propagation reconstructs the entity feature and generates the explicit relation feature (in Fig.\ref{encoder_decoder}(2)). TFP establishes a fast, scalable, and theoretical decoding solution.

Our evaluation of TFP across both translation- and GNN-based encoders, incorporating six advanced EA methods, demonstrates performance gains across widely recognized public datasets. Remarkably, TFP enhances even state-of-the-art (SOTA) methods with minimal additional computational time, less than 6 seconds, setting a new benchmark for efficiency and adaptability in EA strategies. 

Our contributions are multifaceted: 

\noindent
(1) To comprehensively represent the KG structure, we generalize adjacency matrices through multi-view matrices for entity-to-entity, entity-to-relation, relation-to-entity, and relation-to-triple relationships. 

\noindent
(2) We design TFP, a fast, general, and efficient decoding approach that combines multi-view matrices to reconstruct the entity embeddings to maximize homophily by minimizing the Dirichlet energy. 

\noindent
(3) We offer a theoretical foundation for TFP through gradient flow theory. TFP is naturally generated through gradient flow of minimizing the Dirichlet energy. 

\noindent
(4) Extensive experiments are conducted to demonstrate TFP can improve upon current EA methods only rely on entity embeddings with minimal computational cost, typically less than 6 seconds.

% \begin{itemize}

% \item \textbf{Innovative EA Decoder:} We introduce TFP, a pioneering decoding strategy that generalizes adjacency matrices in KGs to include multi-view relationships, thereby reconstructing entity embeddings by maximizing homophily via Dirichlet energy minimization. TFP stands out as the first approach to leverage only entity embedding for structural information-based reconstruction, ensuring broad compatibility with diverse graph encoders.

% \item \textbf{Theoretical Analysis:} We offer a theoretical foundation for embedding reconstruction through Dirichlet energy minimization theory, presenting TFP as an exemplification of natural gradient flow. Our analysis interprets TFP's mechanism through the perspective of a graph heat equation, providing detailed discretization and an explicit Euler solution for practical application.

% \item \textbf{Extensive Experiments:} Through extensive experiments on diverse public datasets, TFP has demonstrated its ability to significantly improve upon current state-of-the-art methods and shown adaptability across different modality settings. Notably, these enhancements are achieved with minimal computational overhead, typically less than 6 seconds.
% \end{itemize}

\section{Related Work}
\label{relatedworks}
\textbf{Entity Alignment Encoders.} EA is predominantly viewed as a graph representation learning endeavor, with encoders designed to intricately model KG structures \cite{zeng2021comprehensive}. Encoders fall into two main categories: translation-based, exemplified by TransE \cite{bordes2013translating} and its variants \cite{chen2017multilingual, guo2019learning,pei2019semi}, which emphasize embedding learning strategy; and GNN-based, which leverage a variety of GNN architectures to produce entity embeddings. These include simple GCNs \cite{wang2018cross}, multi-hop and relational GCNs \cite{sun2020knowledge,yu2020generalized}, graph attention networks \cite{zhu2020collective,mao2021boosting,sun2022revisiting}, and self-supervised GCNs \cite{li2022uncertainty}. Innovations extend to semi-supervised \cite{li2022uncertainty,mao2020mraea} and active learning approaches \cite{zeng2021reinforced,liu2021activeea,berrendorf2021active}, plus the integration of additional information like entity attributes or temporal data \cite{sun2017cross,trisedya2019entity,wu2019relation,xu2021time,xu2022time,liu2023unsupervised}. Our TFP decoding strategy aims to reconstruct entity representation via feature propagation, deriving from Dirichlet energy's gradient flow to maximize homophily. TFP, adaptable across graph encoder types, enhances these encoders, including advanced models.

\noindent
\textbf{Entity Alignment Decoders.} EA decoders traditionally aim at pairing entities from different KGs using learned embeddings. The \textit{greedy search} algorithm \cite{ye2019vectorized,shi2019modeling}, though widely used, can produce multiple mappings for a single entity, conflicting with EA's one-to-one mapping principle. Alternative methods like Cross-Domain Similarity Local Scaling (CSLS) \cite{lample2018word} and the deferred acceptance algorithm \cite{roth2008deferred} have sought to refine alignment accuracy. More recent efforts \cite{xu2020coordinated, zhu2021raga} employ global alignment strategies, using algorithms such as the Hungarian \cite{kuhn1955hungarian} or Sinkhorn \cite{cuturi2013sinkhorn}, yet often overlook KGs' distinctive structural properties. DATTI \cite{mao2022effective} advances decoding by utilizing Third-order Tensor Isomorphism to capture structural information but is limited by its requirement for explicit relation embeddings, absent in many GNN-based models. TFP distinguishes itself by focusing on entity embeddings, fundamental to all encoder types, and establishes a fast and scalable decoding solution requiring only about 6 seconds, applicable across a wide range of encoder architectures.

\section{Preliminary}
\label{Preliminary}
Knowledge graph (KG), $\mathcal{G} = (\mathcal{E}, \mathcal{R}, \mathcal{T})$, stores the real-world knowledge in the form of $\mathcal{T}$, given a set of entities $\mathcal{E}$, a set of relations $\mathcal{R}$, and a set of triples $\mathcal{T} = \{ (h, r, t),| h, t \in \mathcal{E}, r \in \mathcal{R} \}$, where $h, t$ denote the head entity and the tail entity, $r$ denotes the relation.

\noindent
\newtheorem{myDef}{\noindent\bf Definition}
\begin{myDef}
\textbf{Entity Alignment (EA)} aims to discover a one-to-one mapping $\Phi = \{ (e_s, e_t) ,|, e_s \in \mathcal{E}_s, e_t \in \mathcal{E}_t, e_s \equiv e_t \}$ between entities from a source KG $\mathcal{G}_s = (\mathcal{E}_s, \mathcal{R}_s, \mathcal{T}_s)$ to a target KG $\mathcal{G}_t = (\mathcal{E}_t, \mathcal{R}_t, \mathcal{T}_t)$, where with $\equiv$ signifying an equivalence relation between entities $e_s$ and $e_t$.
\end{myDef}

For a simple directed graph $G = (V, E)$ with labels $y = \{{y_i}|{i\in V}\}$, where $V$ represents the set of vertices and $E$ the set of edges, the concept of graph homophily is crucial. Homophily indicates the likelihood of connected nodes sharing the same label. Formally, we define the homophily of a graph $G$ as:
\begin{equation}
\mathscr{H}(G) = \mathbb{E}_{i\in V}\left[\frac{|\{j\in \mathcal{N}_i^{in} : y_i = y_j \}|}{|\mathcal{N}_i^{in}|}\right]
\end{equation}
Here, $|\{j\in \mathcal{N}_i^{in} : y_i = y_j \}|$ quantifies the number of neighbors of node $i$ in $V$ sharing the same label $y_i$  \cite{pei2019geom}. A graph $G$ is considered homophilic if $\mathscr{H}(G) \approx 1$ and heterophilic if $\mathscr{H}(G) \approx 0$.

\begin{myDef}
    \textbf{(Laplacian matrix.)} We define the the adjacency matrix of $G$ as  $\mathbf{A}$, the symmetrically normalized adjacency is represented by $\mathbf{\widetilde{A}} = \mathbf{D}^{-\frac{1}{2}} \mathbf{A} \mathbf{D}^{-\frac{1}{2}}$ , and the symmetrically normalized Laplacian matrix as $\mathbf{\mathbf{\Delta}} = \mathbf{I} - \mathbf{\widetilde{A}}$.
\end{myDef}

The concept of Dirichlet energy is often employed to measure the homophily in graph embeddings. We define it as follows:
\begin{myDef}
\textbf{(Dirichlet Energy.)} Given the graph node embedding $\mathbf{X}\in \mathbb{R}^{N\times d}$, the graph homophily can be measured by the Dirichlet energy of $\mathbf{X}$:
    \begin{equation}
    \begin{aligned}
        \mathscr{L}(\mathbf{X}) = tr(\mathbf{X}^\top \mathbf{\Delta} \mathbf{X}) = \frac{1}{2}\sum_{i,j=1}^Na_{i,j}||\frac{\mathbf{X}_i}{\sqrt{\mathbf{D}_{i,i}+1}} - \frac{\mathbf{X}_j}{\sqrt{\mathbf{D}_{j,j}+1}}||^2_2
    \end{aligned}
    \end{equation}
\end{myDef}

\section{Methodology}
\label{method}
In this section, we introduce our proposed decoding approach, Triple Feature Propagation (TFP). Since the symmetrically Laplacian matrix $\mathbf{\Delta}$ is derived from the undirected adjacency, which captures the entity-to-entity structure, we initially focus on gradient flow within this special situation before extending the approach to the generalized adjacency matrix. Before the decoding process, we need to train an encoder to get the entity embedding as the initial feature $\mathbf{X}^{(0)} = Enocder (\mathcal{G})\in \mathbb{R}^{|\mathcal{E}|\times d}$. The process of TFP is shown in Fig.\ref{structure} and Alg. \ref{algorithm1}.

\subsection{Gradient Flow}
Given seed alignment entity features $x_s$, our objective is to reconstruct the features $x_o$ of other entities by minimizing the Dirichlet energy $\mathscr{L}(\mathbf{X})$. We designate the set of seed alignment entities as $\mathcal{E}_s \subseteq \mathcal{E}$, with the remaining entities denoted by $\mathcal{E}_o = \mathcal{E} \backslash \mathcal{E}_s$. The entities are ordered such that:
\begin{equation}
\label{adjacency}
\mathbf{X} = 
\begin{pmatrix}
    \mathbf{x}_s \\
    \mathbf{x}_{o}
\end{pmatrix},
\quad
\mathbf{\Delta} = 
\begin{pmatrix}
    \mathbf{\Delta}_{ss} & \mathbf{\Delta}_{so} \\
    \mathbf{\Delta}_{os} & \mathbf{\Delta}_{oo}
\end{pmatrix}
\end{equation}
The gradient flow of Dirichlet energy results in the graph heat equation  \cite{chung1996spectral}, with seed alignment features remaining stationary. This is expressed as:
\begin{equation}
\label{heatequation}
    \frac{d\mathbf{X}(t)}{dt} = -\nabla_\mathbf{x} \mathscr{L}(\mathbf{X}(t)) = - \mathbf{\Delta} \mathbf{X}(t)
\end{equation}
By integrating the boundary condition $\mathbf{x_s}(t) = \mathbf{x_s}$, the solution to this heat equation provides the required decoding process. The gradient flow for seed alignments is $\mathbf{0}$, leading to the compact expression:
\begin{equation}
\label{heat equation}
\begin{aligned}
\frac{d}{dt}
\begin{pmatrix}
    \mathbf{x}_s(t) \\
    \mathbf{x}_{o}(t)
\end{pmatrix}
= - 
\begin{pmatrix}
    \mathbf{0} & \mathbf{0}\\
        \mathbf{\Delta}_{os} & \mathbf{\Delta}_{oo}\\
\end{pmatrix}
\begin{pmatrix}
    {\mathbf{x}}_s \\
        {\mathbf{x}}_{o}(t)\\
\end{pmatrix}
= 
\begin{pmatrix}
    \mathbf{0} \\
    -\mathbf{\Delta}_{os} \mathbf{x}_s - \mathbf{\Delta}_{oo} \mathbf{x}_{o}(t)
\end{pmatrix}
\end{aligned}
\end{equation}
Given the positive semi-definite nature of the graph Laplacian matrix, Dirichlet energy $\mathscr{L}$ is convex, and its global minimizer is the solution to the gradient equation $\nabla \mathscr{L}(\mathbf{X}(t))=0$. The solution in equation (\ref{heat equation}) can be expressed as:
\begin{equation}
\label{solution}
    \frac{d\mathbf{x}_o(t)}{dt} = -\mathbf{\Delta}_{os}\mathbf{x}_s - \mathbf{\Delta}_{oo}\mathbf{x}_{o}(t) = \mathbf{0}
\end{equation}
Therefore, we present the following proposition:
\theoremstyle{plain}
\newtheorem{myTh}{\bf Proposition}
\begin{myTh}
\label{existence}
\textbf{(Existence of the solution.)} The matrix $\Delta_{oo}$ is non-singular, allowing the reconstruction of other entity features $\mathbf{x}_o$ using seed alignment entity features $\mathbf{x}_s$ as ${\mathbf{x}}_{o}(t) = -\mathbf{\Delta}^{-1}_{oo}\mathbf{\Delta}_{os}{\mathbf{x}}_s$.
\end{myTh}
\begin{proof}
Please refer to Appendix A.
\end{proof}

However, solving this system of linear equations is computationally intensive, with a complexity of $O(|\mathcal{E}_o|^3)$ for matrix inversion, rendering it impractical for large graphs. Therefore, we discretize the process to tackle this challenge.

\noindent
\textbf{Remark.} This special gradient flow in TFP is generalized to directed graphs, distinguishing it from FP \cite{rossi2022unreasonable} on undirected graphs.

\subsection{Discretization Strategy}
To tackle the computational challenges, we discretize the heat equation (\ref{heat equation}) and adopt an iterative approach for its resolution. By approximating the temporal derivative as a forward difference and discretizing time $t$ with fixed steps ($t = hk$ for step size $h > 0$ and $k = 1, 2, \dots$), we employ the implicit Euler scheme:
\begin{equation}
    \mathbf{X}^{(k+1)} = \mathbf{X}^{(k)} + h \frac{d\mathbf{X}(t)}{dt} \mathbf{X}^{(k)}
\end{equation}
Incorporating the heat equation (\ref{heatequation}), the explicit Euler scheme is subsequently defined as:
\begin{equation}
\mathbf{X}^{(k+1)}
= \begin{pmatrix}
    \mathbf{I} & \mathbf{0} \\
        - h\mathbf{\Delta}_{os} & \mathbf{I} - h\mathbf{\Delta}_{oo}
\end{pmatrix}
    \mathbf{X}^{(k)}
\end{equation}
When focusing on the special case of $h = 1$, we can use the following observation to rewrite the iteration formula.
\begin{equation}
\label{iterformula}
    \mathbf{\widetilde{A}} = \mathbf{I} - \mathbf{\mathbf{\Delta}} = 
    \begin{pmatrix}
        \mathbf{I} - \mathbf{\Delta}_{ss} & - \mathbf{\Delta}_{so}\\
        - \mathbf{\Delta}_{os} & \mathbf{I} - \mathbf{\Delta}_{oo}\\
    \end{pmatrix}
\end{equation}
\begin{equation}
\label{itersolution}
    \mathbf{X}^{(k+1)} = 
    \begin{pmatrix}
        \mathbf{I} & \mathbf{0}\\
        \mathbf{\widetilde{A}}_{os} & \mathbf{\widetilde{A}}_{oo}\\
    \end{pmatrix}
    \mathbf{X}^{(k)}
\end{equation}
This Euler scheme serves as a gradient descent for the Dirichlet energy, reducing the energy and smoothing the features. The approximation of this solution in this case is formalized as follows:
\begin{myTh}
\textbf{(Approximation of the solution.)} 
With the iterative reconstruction solution as delineated in equation (\ref{itersolution}), and considering a sufficiently large iteration count $N$, the entity features will approximate the results of feature propagation as follows:
\begin{equation}
    \mathbf{X}^{(N)} \approx 
    \begin{pmatrix}
        \mathbf{x}_{s} \\ \mathbf{\Delta}^{-1}_{oo}\mathbf{\widetilde{A}}_{os}{\mathbf{x}}_s\\
    \end{pmatrix}
\end{equation}
\end{myTh}
\begin{proof}
Please refer to Appendix B.
\end{proof}
The update process in Equation (\ref{itersolution}) equates to initially multiplying the entity features $\mathbf{X}$ by the matrix $\mathbf{\widetilde{A}}$, followed by resetting the seed alignment features to their original values. This results in a highly scalable and straightforward strategy for reconstructing other entity features:
\begin{equation}
    \mathbf{X}^{(k+1)} \leftarrow \mathbf{\widetilde{A}}\mathbf{X}^{(k)};\quad \mathbf{x}^{(k+1)}_s \leftarrow \mathbf{x}^{(k)}_s
\end{equation}

\subsection{Generalized Adjacency Matrices}
\label{Generalized Adjacency Matrices}

\begin{figure*}[t!]
\centering
\includegraphics[width = \textwidth]{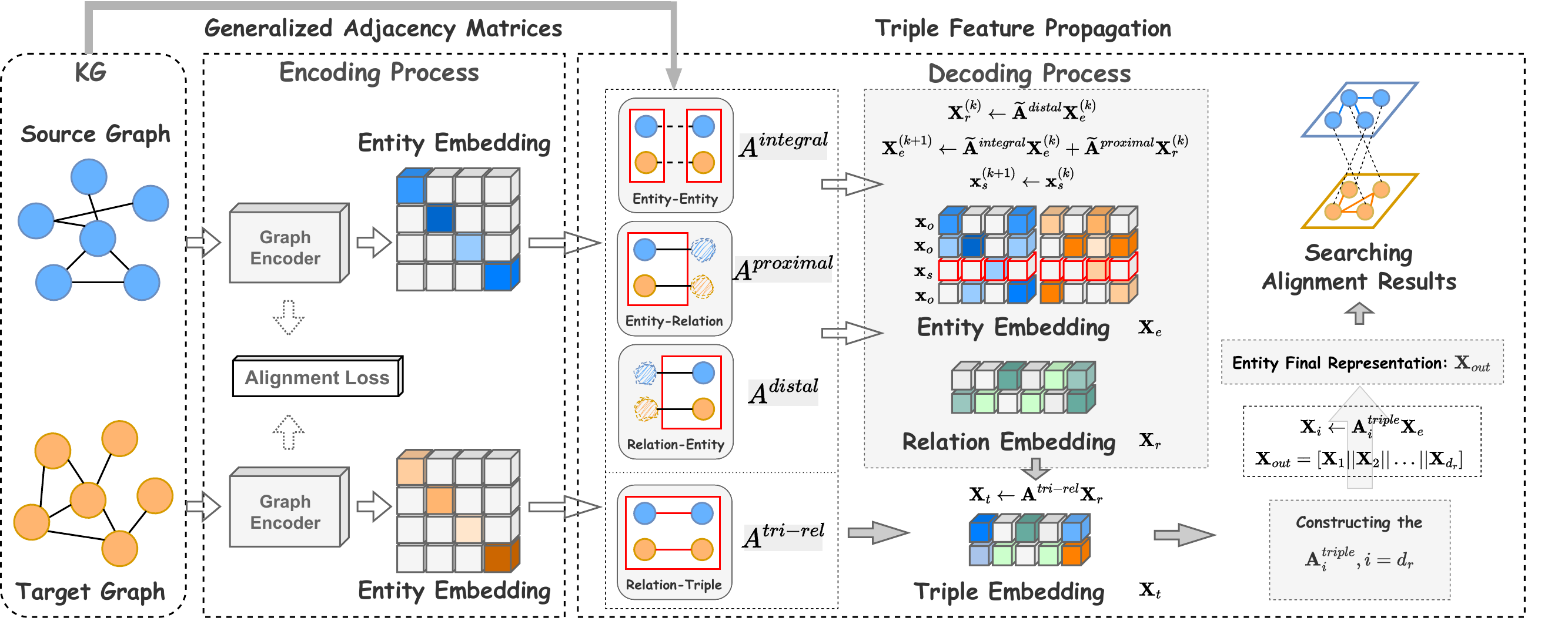}
\caption{The illustration of Triple Feature Propagation.} 
\label{structure}
\end{figure*}

Our discussion has thus far established that the gradient flow of Dirichlet energy facilitates feature reconstruction within basic graph structures. However, KGs extend beyond the realm of simple directed graphs represented by an adjacency matrix $\mathbf{A}$. KGs encompass a richer structure, with triples $\mathcal{T} = \{ (h, r, t) ,| h, t \in \mathcal{E}, r \in \mathcal{R} \}$, highlighting entity-to-entity $(h-t)$, entity-to-relation $(h-r)$, relation-to-entity $(r-t)$, and relation-to-triple $(r-\mathcal{T})$. These perspectives encapsulate key meta-structures in KGs:
\begin{enumerate}
    \item \textbf{Entity-to-Relation.} 
    Beyond the conventional adjacency matrix, the $(h, r)$ pairs in triples provide a unique structural perspective, indicating directional links from head entity $h$ to relation $r$. This directionality is fundamental and unidirectional. For instance, in the triple $(London, CapitalOf, England)$, the entity $[England]$ is connected to $[London]$ via the relation $[CapitalOf]$, but not in reverse. Similar to the adjacency matrix $\mathbf{A}$, we define an $entity\ to\ relation$ matrix $\mathbf{A}^{proximal}\in \mathbb{R}^{|\mathcal{E}|\times |\mathcal{R}|}$ based on the $(h, r)$ pairs as:
    \begin{equation}
    \label{entitytorelation}
    \begin{aligned}
        \exists e \in \mathcal{E}: \mathbf{A}^{proximal}_{i,j} &= 1,\ (h_i,r_j,e)\in \mathcal{T} \quad \mathbf{A}^{proximal}_{i,j} &= 0,\ (h_i,r_j,e)\notin \mathcal{T}
    \end{aligned}
    \end{equation}

    \item \textbf{Relation-to-Entity.} This category, akin to \textit{entity to relation}, encompasses directed and irreversible information. It underscores the more distant positional relationships between relations and tail entities. The resulting structural matrix, $\mathbf{A}^{distal}\in \mathbb{R}^{|\mathcal{R}|\times |\mathcal{E}|}$, is defined as:
    \begin{equation}
    \begin{aligned}
        \exists e \in \mathcal{E}: \mathbf{A}^{distal}_{i,j} &= 1,\ (e, r_i,t_j)\in \mathcal{T} \quad \quad \mathbf{A}^{distal}_{i,j} &= 0,\ (e, r_i,t_j)\notin \mathcal{T}
    \end{aligned}
    \end{equation}

    \item \textbf{Entity-to-Entity.} Representing the fundamental structural category, this is similar to the adjacency matrix $\mathbf{A}$. While $\mathbf{A}$ provides a solid foundation and is rich in spectral properties, it is insufficient to fully represent KG structures. For example, the triples $(London, CapitalOf, England)$ are structurally distinct with $(London, LocateIn, England)$ in KGs, yet a traditional adjacency matrix treats them identically as they both represent $London$ connects $English$, representing $\mathbf{A}_{i,j}=1$. To better represent relationships between entities, we extend the adjacency matrix to $\mathbf{A}^{integral}\in \mathbb{R}^{|\mathcal{E}|\times |\mathcal{E}|}$:
    \begin{equation}
    \begin{aligned}
        \exists r\in \mathcal{R}: \quad \quad \mathbf{A}^{integral}_{i,i} &= |\mathcal{T}_{e_i}| \\ \mathbf{A}^{integral}_{i,j} = |\mathcal{T}_{(h_i,t_j)}|, (h_i,r,t_j) \in \mathcal{T} \quad \quad&\mathbf{A}^{integral}_{i,j} = 0, (h_i,r,t_j) \notin \mathcal{T}
    \end{aligned}
    \end{equation}
    $|\mathcal{T}_{e_i}|$ denotes the number of triples involving entity $e_i$, and $|\mathcal{T}_{(h_i,t_j)}|$ indicates the count of triples with the pair $(h_i,t_j)$.

    \item \textbf{Relation-to-Triple.} As discussed before, the knowledge in KG is stored in the triple $\mathcal{T} = \{ (h, r, t),| h, t \in \mathcal{E}, r \in \mathcal{R} \}$, which is the core representation of the KG structure. The unique role of relation $r$ within triple $\mathcal{T}$ motivates the extension of the adjacency matrix to the triple level through relation $r$. We introduce the adjacency matrix $\mathbf{A}^{tri-rel}\in \mathbb{R}^{|\mathcal{T}|\times |\mathcal{R}|}$, defined as:
    \begin{equation}
    \begin{aligned}
        \mathbf{A}^{tri-rel}_{i,j} &= 1,\ r_j\in \mathcal{T}_i \quad \quad 
        \mathbf{A}^{tri-rel}_{i,j} &= 0,\ r_j\notin \mathcal{T}_i
    \end{aligned}
    \end{equation}

    \textbf{Note.} The \textit{entity-to-triple} matrix is also considered. This matrix is considered to represent \textit{entity-entity-triple} relationships for head $(h)$ and tail $(t)$ entities connections in $\mathcal{T}$. However, we find that these connections are already encapsulated within the \textit{entity-to-entity} matrix $\mathbf{A}^{integral}$. This matrix, $\mathbf{A}^{integral}$, which details connections between head $h$ and tail $t$ entities, effectively doubles as an \textit{entity-to-triple} representation.
\end{enumerate}
Having established that the gradient flow of Dirichlet energy facilitates reconstruction based on common adjacency relationships as outlined in equation (\ref{itersolution}), we extend this concept to the specialized adjacency matrices of KGs. We assume that the normalized form of $\mathbf{A}^{proximal}$, $\mathbf{A}^{distal}$, and $\mathbf{A}^{integral}$, which are similar to graph Laplacian, representing different facets of KG structure, possess analogous spectral properties. The experimental results confirm our view. Consequently, the gradient flow should yield comparable solutions across these matrices. We denote their normalized form $\mathbf{\widetilde{A}}^{k,k=\{proximal,distal,integral\}}$ 
\begin{equation}
\mathbf{\widetilde{A}}^{k}_i=\frac{\mathbf{A}^{k}_i}{\sum_{j=0}^{number of column}\mathbf{A}^{k}_{i,j}}
\end{equation}

\subsection{Triple Feature Propagation}
\label{SectionTFP}

Utilizing the gradient flow of Dirichlet energy on these categories of KG structure, we derive a natural and straightforward generalized propagation strategy. The propagation process is articulated through the following equations:
\begin{equation}
\label{relationequation}
\mathbf{X}_r^{(k)}=\mathbf{\widetilde{A}}^{distal}\mathbf{X}_e^{(k)}
\end{equation}
\begin{equation}
\label{entityequation}
\mathbf{X}_e^{(k+1)}= \mathbf{\widetilde{A}}^{integral}\mathbf{X}_e^{(k)} + \mathbf{\widetilde{A}}^{proximal}\mathbf{X}_r^{(k)}
\end{equation}
\begin{equation}
    \mathbf{x}^{(k+1)}_s =\mathbf{x}^{(k)}_s
\end{equation}
Here, $\mathbf{X}^{(i)}_e, 1\leq i\leq k+1,$ denotes the entity features, initially derived from the graph encoder as $\mathbf{X}_e^{(0)}$. $\mathbf{X}_r$ represents the implicit relation features generated through the propagation strategy. To comprehensively capture the evolving entity and relation information across iterations, we aggregate entity features from each step through concatenation, resulting in the entity feature as follows:
\begin{equation}
    \mathbf{X}_e = [\mathbf{X}^{(0)}||\mathbf{X}^{(1)}||\dots||\mathbf{X}^{(k)}]
\end{equation}

We notice that our process explicitly captures the implicit relation feature $\mathbf{X}_r$, serving as an essential intermediary. By leveraging the \textit{triple-to-relation} matrix, we further refine $\mathbf{X}_r$. We use a hyper-sphere independent random projection to reduce the dimension of $\mathbf{X}_r$ to $d_r$, following the method outlined in \cite{mao2022lightea}. This step enables the creation of the triple feature $\mathbf{X}_t$.
\begin{equation}
    \mathbf{X}_r = ranp(\mathbf{X}_r)\in \mathbb{R}^{|\mathcal{R}|\times d_r}
\end{equation}
\begin{equation}
\label{tripleequation}
    \mathbf{X}_t = \mathbf{{A}}^{tri-rel}\mathbf{X}_r
\end{equation}
Our method enables the representation of the KG structure as a three-dimensional tensor $\mathbf{A}^{triple}\in \mathbb{R}^{|\mathcal{E}|\times |\mathcal{E}|\times d_r}$, rather than the traditional adjacency matrix. The tensor is defined as:
\begin{equation}
\begin{aligned}
    \mathbf{A}^{triple}_{h,t} & = \mathbf{X}_t(h,r,t), (h,t)\in \mathcal{T} \quad \quad \mathbf{A}^{triple}_{h,t} & = \mathbf{0}, (h,t)\notin \mathcal{T}
\end{aligned}
\end{equation}
Here, $\mathbf{X}_t(h,r,t)$ represents the feature of triple $(h,r,t)$ in $\mathbf{X}_t$. The KG structure is encapsulated in $d_r$ slices as $\mathbf{A}^{triple}_1,\dots,\mathbf{A}^{triple}_{d_r} \in \mathbb{R}^{|\mathcal{E}|\times |\mathcal{E}|}$. Parallel to the relation feature, the entity feature $\mathbf{X}_e$ is also scaled through hyper-sphere independent random projection to $d_e$ dimension. The propagation process is based on the final entity features $\mathbf{X}_e$, allowing for the continuation of gradient flow from the triple perspective. Since each $\mathbf{A}^{triple}_i$ represents a segmented view of the overall $\mathbf{A}^{triple}$, their concatenation is essential to compile the comprehensive final feature.
\begin{equation}
\label{fianlequation}
    \mathbf{X}_e = ranp(\mathbf{X}_e)\in \mathbb{R}^{|\mathcal{R}|\times d_e},\quad \mathbf{X}_i = \mathbf{A}^{triple}_i\mathbf{X}_e, i=1,\dots,d_r
\end{equation}
\begin{equation}
\label{concatenate}
    \mathbf{X}_{out} = [\mathbf{X}_1||\mathbf{X}_2||\dots||\mathbf{X}_{d_r}]
\end{equation}
TFP's implementation is outlined in Algorithm 1. To identify alignment results, we adopt an approach from \cite{mao2022lightea,mao2022effective} that frames the search as an assignment problem, moving beyond traditional Euclidean or cosine similarity calculations. This method upholds the one-to-one matching constraint and facilitates the use of the Sinkhorn operator \cite{cuturi2013sinkhorn} for faster computation. Detailed procedures are described in Appendix C.

\noindent
\begin{minipage}{.49\textwidth}
\centering
\vspace{-0.7cm}
\begin{algorithm}[H]
\scriptsize
\caption{\small TFP}
\label{algorithm1}
\renewcommand{\algorithmicrequire}{\textbf{Input:}}
\renewcommand{\algorithmicensure}{\textbf{Output:}}
\setlength{\algorithmicindent}{0em} % Add this line to remove indentation
\begin{algorithmic}[1]        
\REQUIRE The entity embedding $\mathbf{X}^{(0)}$, the triples $\mathcal{T}$, iteration number $K$, dimension $d_r,d_e$.
\ENSURE The reconstructed entity feature  $\mathbf{X}_{out}$.  %%output
\STATE Initialize $\mathbf{\widetilde{A}}^{proximal}$, $\mathbf{\widetilde{A}}^{distal}$, $\mathbf{\widetilde{A}}^{integral}$, and $\mathbf{A}^{tri-rel}$ through the triples $\mathcal{T}$.
\FOR{$k=1\rightarrow K$}
\STATE $\mathbf{X}_r^{(k)} \leftarrow \mathbf{\widetilde{A}}^{distal}\mathbf{X}_e^{(k)}$
\STATE $\mathbf{X}_e^{(k+1)} \leftarrow \mathbf{\widetilde{A}}^{integral}\mathbf{X}_e^{(k)} + \mathbf{\widetilde{A}}^{proximal}\mathbf{X}_r^{(k)}$
\STATE $\mathbf{x}^{(k+1)}_s \leftarrow \mathbf{x}^{(k)}_s$
\ENDFOR
\STATE $\mathbf{X}_r \leftarrow random\_projection(\mathbf{X}_r^{(K)}, d_r)$
\STATE $\mathbf{X}_t \leftarrow \mathbf{{A}}^{tri-rel}\mathbf{X}_r$
\STATE Generate $\mathbf{A}^{triple}_i$ through $\mathbf{X}_t$.
\STATE $\mathbf{X}_e \leftarrow random\_projection(\mathbf{X}_e^{(K)}, d_e)$
\STATE $\mathbf{X}_i \leftarrow \mathbf{A}^{triple}_i\mathbf{X}_e^{(K)}$
\STATE $\mathbf{X}_{out} \leftarrow [\mathbf{X}_1||\mathbf{X}_2||\dots||\mathbf{X}_{d_r}]$
\STATE \textbf{return} $\mathbf{X}_{out}$
\end{algorithmic}
\end{algorithm}
\end{minipage}
\hfill
\begin{minipage}{.49\textwidth}
\centering
\small
\captionof{table}{\small Statistics of the DBP15K and SRPRS datasets, highlighting SRPRS as sparse KGs typical of real-world scenarios.}
\label{datasets}
\centering
\resizebox{\textwidth}{!}{
\begin{tabular}{cc|ccc}
\toprule
\multicolumn{2}{c|}{Datasets}                & Entity & Relation & Triple \\
\midrule
\multirow{2}{*}{DBP$_{ZH-EN}$}   & \multicolumn{0}{|c|}{Chinese}  & 19388         & 1701            & 70414         \\
                                  & \multicolumn{0}{|c|}{English}  & 19572         & 1323            & 95142         \\ \midrule
\multirow{2}{*}{DBP$_{JA-EN}$}   & \multicolumn{0}{|c|}{Japaense}  & 19814         & 1299            & 77214         \\
                                  & \multicolumn{0}{|c|}{English}   & 19780         & 1153            & 93484         \\ \midrule
\multirow{2}{*}{DBP$_{FR-EN}$}   & \multicolumn{0}{|c|}{French}   & 19661         & 903             & 105998        \\
                                  & \multicolumn{0}{|c|}{English}  & 19993         & 1208            & 115722        \\ \midrule
\multirow{2}{*}{SRPRS$_{FR-EN}$} & \multicolumn{0}{|c|}{French}   & 15000         & 177             & 33532         \\
                                  & \multicolumn{0}{|c|}{English}  & 15000         & 221             & 36508         \\ \midrule
\multirow{2}{*}{SRPRS$_{DE-EN}$} & \multicolumn{0}{|c|}{German}   & 15000         & 120             & 37377         \\
                                  & \multicolumn{0}{|c|}{English}  & 15000         & 222             & 38363         \\ \bottomrule
\end{tabular}}
\end{minipage}%

% \centering
% \captionof{table}{\small Execution time (s) of EA methods decoding with TFP, where TFP(C) / (G) denotes CPU / GPU execution.}
% \label{Timeresult}
% \resizebox{\textwidth}{!}{
% \begin{tabular}{ccc|ccc}
% \toprule
% \multicolumn{3}{c|}{Translation-based} & \multicolumn{3}{c}{GNN-based} \\
% \midrule
% \multicolumn{1}{c|}{Time/s}    & DBP15K & SRPRS & \multicolumn{1}{c|}{Time/s}   & DBP15K & SRPRS \\ \midrule
% \multicolumn{1}{c|}{AlignE}    & 2087   & 1190     & \multicolumn{1}{c|}{MRAEA}    & 1743   & 558      \\
% \rowcolor{gray!20}
% \multicolumn{1}{c|}{TFP(C)}  & 13.9   & 8.1      & \multicolumn{1}{c|}{TFP(C)} & 16.6   & 10.6     \\
% \rowcolor{gray!20}
% \multicolumn{1}{c|}{TFP(G)}  & 4.8    & 4.2      & \multicolumn{1}{c|}{TFP(G)} & 5.9    & 4.6      \\ \midrule
% \multicolumn{1}{c|}{RSN}       & 3659   & 1279     & \multicolumn{1}{c|}{RREA}     & 323    & 276      \\
% \rowcolor{gray!20}
% \multicolumn{1}{c|}{TFP(C)}  & 14.2   & 9.2      & \multicolumn{1}{c|}{TFP(C)} & 16.3   & 11       \\
% \rowcolor{gray!20}
% \multicolumn{1}{c|}{TFP(G)}  & 4.8    & 3.7      & \multicolumn{1}{c|}{TFP(G)} & 5.7    & 3.8      \\ \midrule
% \multicolumn{1}{c|}{TransEdge} & 1625   & 907      & \multicolumn{1}{c|}{DualAMN}  & 177    & 163      \\
% \rowcolor{gray!20}
% \multicolumn{1}{c|}{TFP(C)}  & 12.9   & 8.7      & \multicolumn{1}{c|}{TFP(C)} & 17.7   & 11.4     \\
% \rowcolor{gray!20}
% \multicolumn{1}{c|}{TFP(G)}  & 4.7    & 3.6      & \multicolumn{1}{c|}{TFP(G)} & 5.1    & 4.6      \\
% \bottomrule
% \end{tabular}}

\section{Experiments}
\label{experiments}
To thoroughly assess the proposed Triple Feature Propagation method, our experiments target the following research questions:

\begin{enumerate}
    \item Can TFP effectively generalize across different graph encoders and integrate with structural-based state-of-the-art methods?
    \item How does TFP's performance evolve across iterations when applied to different encoders on various datasets?
    \item How does TFP's time efficiency compare to that of encoders?
    \item How does TFP perform when additional information is integrated, compared to other textual EA methods?
\end{enumerate}

\subsection{Experimental Settings}
Our experimental setup includes details on datasets, baselines and evaluation metrics, and implementation details.

\noindent
\textbf{Datasets:}\quad
We utilize two widely recognized datasets to test our decoding algorithm: \textbf{(1)} DBP15K  \cite{sun2017cross} comprises three cross-lingual subsets from multilingual DBpedia. Each subset contains 15,000 entity pairs. \textbf{(2)} SRPRS  \cite{guo2019learning}. Similar to DBP15K in terms of the number of entity pairs but with fewer triples. For consistency with prior research  \cite{wang2018cross,chen2017multilingual}, we use a 30/70 split of pre-aligned entity pairs for training and testing encoders, respectively. Table \ref{datasets} provides the details.

% Results are averaged over five independent runs. The more details have been provided in Appendix D.

\noindent
\textbf{Baseline and Evaluation Metrics:}\quad
\textbf{(1)} Our evaluation includes six prominent EA encoder methods, divided into GNN-based—MRAEA \cite{mao2020mraea}, RREA \cite{mao2020relational}, Dual-AMN \cite{mao2021boosting} (with Dual-AMN as the structural SOTA)—and translation-based—AlignE \cite{sun2018bootstrapping}, RSN \cite{guo2019learning}, and TransEdge \cite{sun2019transedge} (with TransEdge as the premier in its class). 
\textbf{(2)} For decoding baselines, we benchmark against the Hungarian algorithm \cite{kuhn1955hungarian}, which is the prominent solution in previous work \cite{xu2020coordinated}, and the SOTA of EA decoder, DATTI \cite{mao2022effective}.
\textbf{(3)} Following most previous works \cite{ref_article31, zeng2021comprehensive}, our evaluation employs cosine similarity for EA and H@k and MRR metrics for a thorough evaluation. Details about these metrics can be found in Appendix D.

\noindent
\textbf{Note:} In our primary experiments, we utilize six encoders that only focus on structural information—the core information of EA tasks. Although some recent methods outperform these by integrating additional information, such as textual attributes, we also explore these situations in Section \ref{Q4} for a broader evaluation. Comparative methods are detailed in Table.\ref{additionalresults}.

\noindent
\textbf{Implementaion Details:}\quad
\label{Implementaion Details}
The output dimensions $d$ and other hyper-parameters of all encoders adhere to their original settings in their papers: Dual-AMN $(d=768)$, RREA$(d=600)$, MRAEA $(d=600)$, AlignE $(d=75)$, RSN $(d=256)$, and TransEdge $(d=75)$. The iteration $k$ is set to their best results, which details described in section \ref{iteration}. Other hyper-parameters remain the same for all datasets and methods: relation scale dimension $d_r=512$, entity scale dimension $d_e=16$. All experiments are conducted on a PC with an NVIDIA RTX A6000 GPU and an Intel Xeon Gold 6248R CPU.

\subsection{Main Results (Q1)}

The primary experimental results are summarized in Table \ref{mainresults}. Among the evaluated six encoders, Dual-AMN demonstrates superior performance across all datasets, underscoring the efficacy of GNN encoders. Interestingly, TransEdge shows notable success on the DBP15K dataset but underperforms on SRPRS. This can be attributed to TransEdge's reliance on existing edge semantic information for entity dependency capture, a feature less prevalent in the sparser KGs like SRPRS. Conversely, GNN-based models exhibit robust performance on sparse graphs, highlighting their aptitude for handling sparse topologies.
\vspace{-12pt}
\subsubsection{\textbf{GNN-based encoders:}}
Our TFP approach consistently enhances performance across datasets, including notable improvements with the SOTA structure-focused Dual-AMN method. Specifically, TFP achieves a 2.25\% increase in Hits@1 and a 2.02\% increase in MRR on DBP15K$_{FR-EN}$, and a 1.28\% increase in Hits@1 and a 1.19\% increase in MRR on SRPRS$_{FR-EN}$. The Hungarian particularly benefits MRAEA, probably due to Dual-AMN and RREA's bi-directional relation strategy already incorporating the core idea of the Hungarian \cite{mao2022effective}. Unlike the DATTI, which requires both entity and relation embeddings, TFP can be applied to any encoder, even those that only produce entity embeddings. These findings underscore TFP’s ability to: (\romannumeral1) effectively reconstruct entity embeddings via the gradient flow of Dirichlet energy, facilitating feature propagation and homophily; (\romannumeral2) demonstrate effectiveness and adaptability across various GNN-based graph encoders; (\romannumeral3) enhance structure-focused methods.

\vspace{-12pt}
\subsubsection{\textbf{Translation-based encoders:}}
TFP significantly boosts the performance of translation-based EA encoders, effectively capturing multi-view structural information through the gradient flow to maximize the homophily. Notably, TFP improves AlignE's performance by 28.87\% in Hits@1 on the DBP15K$_{JA-EN}$ dataset and by 25.8\% on the SRPRS$_{DE-EN}$ dataset. For the best translation-based encoder, TransEdge, TFP still achieves substantial enhancements, including at least a 7.01\% increase in Hits@1 on DBP15K and 15.22\% on SRPRS. These findings demonstrate TFP's effectiveness and general applicability to translation-based encoders. Particularly, TFP's significant improvement on the SRPRS dataset indicates its capability to bridge gaps in sparse KG topology understanding, an area where translation-based EA encoders typically struggle.

\begin{table*}[t]
\caption{\footnotesize Main experimental results on DBP15K and SRPRS datasets. All results and initial embeddings were derived using official code with default hyper-parameters. "Improv." indicates the percentage improvement of TFP over EA encoder results. The Hungarian algorithm (Hun) produces only one aligned pair per entity, thus only Hits@1 is reported. DATTI cannot decode MRAEA and RREA, as it requires both entity and relation embeddings, whereas these encoders provide only entity embeddings.}
\label{mainresults}
\setlength{\tabcolsep}{3pt}
\centering
\huge
\resizebox{\textwidth}{!}{
\begin{tabular}{cc|ccccccccccccccc}
\toprule
\multicolumn{2}{c|}{Datasets}     & \multicolumn{3}{c}{DBP$_{FR-EN}$}                & \multicolumn{3}{c}{DBP$_{JA-EN}$}                  & \multicolumn{3}{c}{DBP$_{ZH-EN}$}                  & \multicolumn{3}{c}{SRPRS$_{FR-EN}$}               & \multicolumn{3}{c}{SRPRS$_{DE-EN}$}                \\
 \cmidrule(lr){1-2}\cmidrule(lr){3-5}\cmidrule(lr){6-8}\cmidrule(lr){9-11}\cmidrule(lr){12-14} \cmidrule(lr){15-17}
\multicolumn{2}{c|}{Model}                                                 & H@1              & H@10            & MRR              & H@1              & H@10             & MRR              & H@1              & H@10             & MRR              & H@1              & H@10             & MRR              & H@1              & H@10             & MRR              \\
\midrule
\multicolumn{1}{c|}{\multirow{15}{*}{\rotatebox{90}{GNN-based}}}         & MRAEA            & 71.63            & 94.28           & 80.02            & 68.7             & 93.19            & 77.66            & 68.7             & 92.94            & 77.43            & 43.44            & 75.1             & 53.85            & 56.18            & 51.47            & 64.87            \\
\multicolumn{1}{c|}{}    & +Hun             & 80.33            & -               & -                & 76.54            & -                & -                & 78.04            & -                & -                & 45.29            & -                & -                & 58.70            & -                & -                \\
\multicolumn{1}{c|}{}    & +DATTI           & -                & -               & -                & -                & -                & -                & -                & -                & -                & -                & -                & -                & -                & -                & -                \\
\multicolumn{1}{c|}{}                               &\cellcolor[HTML]{EFEFEF}\textbf{+TFP}              &\cellcolor[HTML]{EFEFEF} 75.45            &\cellcolor[HTML]{EFEFEF} 95.87           &\cellcolor[HTML]{EFEFEF} 83.04            &\cellcolor[HTML]{EFEFEF} 74.26            &\cellcolor[HTML]{EFEFEF} \cellcolor[HTML]{EFEFEF}94.98            &\cellcolor[HTML]{EFEFEF} 82.08            &\cellcolor[HTML]{EFEFEF} 74.4             &\cellcolor[HTML]{EFEFEF} 94.54            &\cellcolor[HTML]{EFEFEF} 81.82            &\cellcolor[HTML]{EFEFEF} 46.06            &\cellcolor[HTML]{EFEFEF} 75.55            &\cellcolor[HTML]{EFEFEF} 55.89            &\cellcolor[HTML]{EFEFEF} 58.83            &\cellcolor[HTML]{EFEFEF} 81.89            &\cellcolor[HTML]{EFEFEF} 66.78            \\
\multicolumn{1}{c|}{}                                   &\cellcolor[HTML]{EFEFEF} \textbf{Improv.} &\cellcolor[HTML]{EFEFEF} \textbf{5.3\%}  &\cellcolor[HTML]{EFEFEF} \textbf{1.7\%} &\cellcolor[HTML]{EFEFEF} \textbf{3.8\%}  &\cellcolor[HTML]{EFEFEF} \textbf{8.1\%}  & \cellcolor[HTML]{EFEFEF}\textbf{1.9\%}  & \cellcolor[HTML]{EFEFEF}\textbf{5.7\%}  & \cellcolor[HTML]{EFEFEF}\textbf{8.3\%}  &\cellcolor[HTML]{EFEFEF} \textbf{1.7\%}  & \cellcolor[HTML]{EFEFEF}\textbf{5.7\%}  &\cellcolor[HTML]{EFEFEF} \textbf{6.0\%}  &\cellcolor[HTML]{EFEFEF} \textbf{0.6\%}  &\cellcolor[HTML]{EFEFEF} \textbf{3.8\%}  &\cellcolor[HTML]{EFEFEF} \textbf{4.7\%}  & \cellcolor[HTML]{EFEFEF}\textbf{59.1\%} &\cellcolor[HTML]{EFEFEF} \textbf{2.9\%}  \\ \cmidrule{2-17}
\multicolumn{1}{c|}{}                                   & RREA             & 73.4             & 94.8            & 81.36            & 70.59            & 94.13            & 79.11            & 70.73            & 93.21            & 78.97            & 42.96            & 73.96            & 53.41            & 56.74            & 81.24            & 65.18            \\
\multicolumn{1}{c|}{}   & +Hun             & 80.42            & -               & -                & 78.44            & -                & -                & 78.69            & -                & -                & 45.68            & -                & -                & 58.96            & -                & -                \\
\multicolumn{1}{c|}{}   & +DATTI           & -                & -               & -                & -                & -                & -                & -                & -                & -                & -                & -                & -                & -                & -                & -                \\           
\multicolumn{1}{c|}{}                                   & \cellcolor[HTML]{EFEFEF}\textbf{+TFP}    & \cellcolor[HTML]{EFEFEF}80.79           & \cellcolor[HTML]{EFEFEF}96.98          & \cellcolor[HTML]{EFEFEF}86.79           & \cellcolor[HTML]{EFEFEF}78.83           & \cellcolor[HTML]{EFEFEF}96.00           & \cellcolor[HTML]{EFEFEF}85.23           & \cellcolor[HTML]{EFEFEF}79.43           & \cellcolor[HTML]{EFEFEF}95.45           & \cellcolor[HTML]{EFEFEF}85.35           & \cellcolor[HTML]{EFEFEF}47.23           & \cellcolor[HTML]{EFEFEF}75.60           & \cellcolor[HTML]{EFEFEF}56.67           & \cellcolor[HTML]{EFEFEF}59.98           & \cellcolor[HTML]{EFEFEF}82.33           & \cellcolor[HTML]{EFEFEF}67.80           \\
                                     
\multicolumn{1}{c|}{}                                   & \cellcolor[HTML]{EFEFEF}\textbf{Improv.} & \cellcolor[HTML]{EFEFEF}\textbf{10.1\%} & \cellcolor[HTML]{EFEFEF}\textbf{2.3\%} & \cellcolor[HTML]{EFEFEF}\textbf{6.7\%}  & \cellcolor[HTML]{EFEFEF}\textbf{11.7\%} & \cellcolor[HTML]{EFEFEF}\textbf{2.0\%}  & \cellcolor[HTML]{EFEFEF}\textbf{7.7\%}  & \cellcolor[HTML]{EFEFEF}\textbf{12.3\%} & \cellcolor[HTML]{EFEFEF}\textbf{2.4\%}  & \cellcolor[HTML]{EFEFEF}\textbf{8.1\%}  & \cellcolor[HTML]{EFEFEF}\textbf{9.9\%}  & \cellcolor[HTML]{EFEFEF}\textbf{2.2\%}  & \cellcolor[HTML]{EFEFEF}\textbf{6.1\%}  & \cellcolor[HTML]{EFEFEF}\textbf{5.7\%}  & \cellcolor[HTML]{EFEFEF}\textbf{1.3\%}  & \cellcolor[HTML]{EFEFEF}\textbf{4.0\%}  \\
                                     
                                     \cmidrule{2-17}
\multicolumn{1}{c|}{}                                   & DualAMN          & 83.43            & 96.19           & 88.14            & 80.31            & 94.69            & 85.57            & 80.39            & 93.68            & 85.29            & 48.28            & 75.51            & 57.34            & 61.2             & 81.91            & 68.3             \\
\multicolumn{1}{c|}{} & +Hun             & 83.87            & -               & -                & 80.39            & -                & -                & 80.12            & -                & -                & 48.32            & -                & -                & 61.15            & -                & -                \\
\multicolumn{1}{c|}{} & +DATTI           & 87.30            & 97.90           & 91.30            & 83.60            & 96.90            & 88.40            & 83.50            & 95.30            & 88.00            & 49.50            & 76.00            & 58.30            & 62.30            & 82.20            & 69.10            \\   
\multicolumn{1}{c|}{}                                  & \cellcolor[HTML]{EFEFEF}\textbf{+TFP}    & \cellcolor[HTML]{EFEFEF}85.31           & \cellcolor[HTML]{EFEFEF}97.26          & \cellcolor[HTML]{EFEFEF}89.92           & \cellcolor[HTML]{EFEFEF}81.26           & \cellcolor[HTML]{EFEFEF}95.85           & \cellcolor[HTML]{EFEFEF}86.82           & \cellcolor[HTML]{EFEFEF}81.64           & \cellcolor[HTML]{EFEFEF}95.24           & \cellcolor[HTML]{EFEFEF}86.84           & \cellcolor[HTML]{EFEFEF}48.90           & \cellcolor[HTML]{EFEFEF}76.28           & \cellcolor[HTML]{EFEFEF}58.02           & \cellcolor[HTML]{EFEFEF}61.84           & \cellcolor[HTML]{EFEFEF}82.61           & \cellcolor[HTML]{EFEFEF}69.01           \\

\multicolumn{1}{c|}{}                            & \cellcolor[HTML]{EFEFEF}\textbf{Improv.} & \cellcolor[HTML]{EFEFEF}\textbf{2.3\%}  & \cellcolor[HTML]{EFEFEF}\textbf{1.1\%} & \cellcolor[HTML]{EFEFEF}\textbf{2.0\%}  & \cellcolor[HTML]{EFEFEF}\textbf{1.2\%}  & \cellcolor[HTML]{EFEFEF}\textbf{1.2\%}  & \cellcolor[HTML]{EFEFEF}\textbf{1.5\%}  & \cellcolor[HTML]{EFEFEF}\textbf{1.6\%}  & \cellcolor[HTML]{EFEFEF}\textbf{1.7\%}  & \cellcolor[HTML]{EFEFEF}\textbf{1.8\%}  & \cellcolor[HTML]{EFEFEF}\textbf{1.3\%}  & \cellcolor[HTML]{EFEFEF}\textbf{1.0\%}  & \cellcolor[HTML]{EFEFEF}\textbf{1.2\%}  & \cellcolor[HTML]{EFEFEF}\textbf{1.0\%}  & \cellcolor[HTML]{EFEFEF}\textbf{0.9\%}  & \cellcolor[HTML]{EFEFEF}\textbf{1.0\%}  \\

\midrule
\multicolumn{1}{c|}{\multirow{15}{*}{\rotatebox{90}{Translation-based}}} & AlignE           & 53.36            & 86.55           & 64.93            & 50.12            & 83.91            & 61.58            & 50.96            & 82.3             & 61.7             & 34.33            & 65.44            & 44.68            & 44.07            & 69.43            & 52.7             \\
\multicolumn{1}{c|}{} & +Hun             & 64.16            & -               & -                & 58.41            & -                & -                & 60.37            & -                & -                & 37.24            & -                & -                & 50.15            & -                & -                \\
\multicolumn{1}{c|}{} & +DATTI           & 58.55            & 85.96           & 68.01            & 55.62            & 82.80            & 64.80            & 57.50            & 82.70            & 66.15            & 39.15            & 69.12            & 49.19            & 53.55            & 75.59            & 61.29            \\                                   
\multicolumn{1}{c|}{}                                   & \cellcolor[HTML]{EFEFEF}\textbf{+TFP}    & \cellcolor[HTML]{EFEFEF}68.45           & \cellcolor[HTML]{EFEFEF}91.96          & \cellcolor[HTML]{EFEFEF}76.81           & \cellcolor[HTML]{EFEFEF}64.59           & \cellcolor[HTML]{EFEFEF}90.77           & \cellcolor[HTML]{EFEFEF}73.76           & \cellcolor[HTML]{EFEFEF}66.98           & \cellcolor[HTML]{EFEFEF}89.78           & \cellcolor[HTML]{EFEFEF}74.89           & \cellcolor[HTML]{EFEFEF}42.11           & \cellcolor[HTML]{EFEFEF}71.80           & \cellcolor[HTML]{EFEFEF}51.98           & \cellcolor[HTML]{EFEFEF}55.44           & \cellcolor[HTML]{EFEFEF}78.61           & \cellcolor[HTML]{EFEFEF}63.48           \\

\multicolumn{1}{c|}{}   & \cellcolor[HTML]{EFEFEF}\textbf{Improv.} & \cellcolor[HTML]{EFEFEF}\textbf{28.3\%} & \cellcolor[HTML]{EFEFEF}\textbf{6.3\%} & \cellcolor[HTML]{EFEFEF}\textbf{18.3\%} & \cellcolor[HTML]{EFEFEF}\textbf{28.9\%} & \cellcolor[HTML]{EFEFEF}\textbf{8.2\%}  & \cellcolor[HTML]{EFEFEF}\textbf{19.8\%} & \cellcolor[HTML]{EFEFEF}\textbf{31.4\%} & \cellcolor[HTML]{EFEFEF}\textbf{9.1\%}  & \cellcolor[HTML]{EFEFEF}\textbf{21.4\%} & \cellcolor[HTML]{EFEFEF}\textbf{22.7\%} & \cellcolor[HTML]{EFEFEF}\textbf{9.7\%}  & \cellcolor[HTML]{EFEFEF}\textbf{16.3\%} & \cellcolor[HTML]{EFEFEF}\textbf{25.8\%} & \cellcolor[HTML]{EFEFEF}\textbf{13.2\%} & \cellcolor[HTML]{EFEFEF}\textbf{20.5\%} \\
                                       \cmidrule{2-17}
\multicolumn{1}{c|}{}                                   & RSN              & 63.17            & 86.37           & 71.33            & 59.13            & 81.5             & 67.03            & 60.67            & 82.86            & 68.53            & 35.1             & 63.78            & 44.73            & 51.07            & 74.43            & 59.02            \\
\multicolumn{1}{c|}{}     & +Hun             & 69.25            & -               & -                & 63.33            & -                & -                & 66.08            & -                & -                & 37.42            & -                & -                & 53.79            & -                & -                \\
\multicolumn{1}{c|}{}     & +DATTI           & 72.00            & 91.80           & 79.00            & 68.60            & 89.50            & 75.90            & 72.10            & 90.30            & 78.50            & 40.70            & 69.40            & 50.20            & 55.90            & 78.20            & 63.70            \\
                                   
\multicolumn{1}{c|}{}                                   & \cellcolor[HTML]{EFEFEF}\textbf{+TFP}    & \cellcolor[HTML]{EFEFEF}77.79           & \cellcolor[HTML]{EFEFEF}92.45          & \cellcolor[HTML]{EFEFEF}83.73           & \cellcolor[HTML]{EFEFEF}73.53           & \cellcolor[HTML]{EFEFEF}91.97           & \cellcolor[HTML]{EFEFEF}80.04           & \cellcolor[HTML]{EFEFEF}75.44           & \cellcolor[HTML]{EFEFEF}91.71           & \cellcolor[HTML]{EFEFEF}81.33           & \cellcolor[HTML]{EFEFEF}44.32           & \cellcolor[HTML]{EFEFEF}73.22           & \cellcolor[HTML]{EFEFEF}53.97           & \cellcolor[HTML]{EFEFEF}59.27           & \cellcolor[HTML]{EFEFEF}81.14           & \cellcolor[HTML]{EFEFEF}66.84           \\
\multicolumn{1}{c|}{}     & \cellcolor[HTML]{EFEFEF}\textbf{Improv.} & \cellcolor[HTML]{EFEFEF}\textbf{23.1\%} & \cellcolor[HTML]{EFEFEF}\textbf{7.0\%} & \cellcolor[HTML]{EFEFEF}\textbf{17.4\%} & \cellcolor[HTML]{EFEFEF}\textbf{24.4\%} & \cellcolor[HTML]{EFEFEF}\textbf{12.8\%} & \cellcolor[HTML]{EFEFEF}\textbf{19.4\%} & \cellcolor[HTML]{EFEFEF}\textbf{24.3\%} & \cellcolor[HTML]{EFEFEF}\textbf{10.7\%} & \cellcolor[HTML]{EFEFEF}\textbf{18.7\%} & \cellcolor[HTML]{EFEFEF}\textbf{26.3\%} & \cellcolor[HTML]{EFEFEF}\textbf{14.8\%} & \cellcolor[HTML]{EFEFEF}\textbf{20.7\%} & \cellcolor[HTML]{EFEFEF}\textbf{16.1\%} & \cellcolor[HTML]{EFEFEF}\textbf{9.0\%}  & \cellcolor[HTML]{EFEFEF}\textbf{13.2\%} \\
                                      \cmidrule{2-17}
\multicolumn{1}{c|}{}                                   & TransEdge        & 76.9             & 93.97           & 83.03            & 74.66            & 92.93            & 81.1             & 76.19            & 92.16            & 81.81            & 40.81            & 67.66            & 49.73            & 55.65            & 75.3             & 64.28            \\
\multicolumn{1}{c|}{}  & +Hun             & 79.55            & -               & -                & 77.06            & -                & -                & 78.72            & -                & -                & 42.65            & -                & -                & 57.41            & -                & -                \\
\multicolumn{1}{c|}{}  & +DATTI           & 81.80            & 96.50           & 87.30            & 80.40            & 95.70            & 86.10            & 81.40            & 94.70            & 86.30            & 44.10            & 70.70            & 52.10            & 59.30            & 78.20            & 67.30            \\
                                    
\multicolumn{1}{c|}{}                                   & \cellcolor[HTML]{EFEFEF}\textbf{+TFP}    & \cellcolor[HTML]{EFEFEF}82.36           & \cellcolor[HTML]{EFEFEF}96.90          & \cellcolor[HTML]{EFEFEF}87.91           & \cellcolor[HTML]{EFEFEF}80.52           & \cellcolor[HTML]{EFEFEF}95.96           & \cellcolor[HTML]{EFEFEF}86.19           & \cellcolor[HTML]{EFEFEF}81.53           & \cellcolor[HTML]{EFEFEF}95.21           & \cellcolor[HTML]{EFEFEF}86.50           & \cellcolor[HTML]{EFEFEF}49.03           & \cellcolor[HTML]{EFEFEF}75.22           & \cellcolor[HTML]{EFEFEF}56.49           & \cellcolor[HTML]{EFEFEF}64.12           & \cellcolor[HTML]{EFEFEF}81.70           & \cellcolor[HTML]{EFEFEF}70.95           \\
\multicolumn{1}{c|}{} & \cellcolor[HTML]{EFEFEF}\textbf{Improv.} & \cellcolor[HTML]{EFEFEF}\textbf{7.1\%}  & \cellcolor[HTML]{EFEFEF}\textbf{3.1\%} & \cellcolor[HTML]{EFEFEF}\textbf{5.9\%}  & \cellcolor[HTML]{EFEFEF}\textbf{7.8\%}  & \cellcolor[HTML]{EFEFEF}\textbf{3.3\%}  & \cellcolor[HTML]{EFEFEF}\textbf{6.3\%}  & \cellcolor[HTML]{EFEFEF}\textbf{7.0\%}  & \cellcolor[HTML]{EFEFEF}\textbf{3.3\%}  & \cellcolor[HTML]{EFEFEF}\textbf{5.7\%}  & \cellcolor[HTML]{EFEFEF}\textbf{20.1\%} & \cellcolor[HTML]{EFEFEF}\textbf{11.2\%} & \cellcolor[HTML]{EFEFEF}\textbf{13.6\%} & \cellcolor[HTML]{EFEFEF}\textbf{15.2\%} & \cellcolor[HTML]{EFEFEF}\textbf{8.5\%}  & \cellcolor[HTML]{EFEFEF}\textbf{10.4\%} \\
 \bottomrule
\end{tabular}}
\end{table*}

\subsection{Exploratory Analysis (Q2)}
\label{iteration}

\begin{figure*}[t!]
\centering
\includegraphics[width = \linewidth]{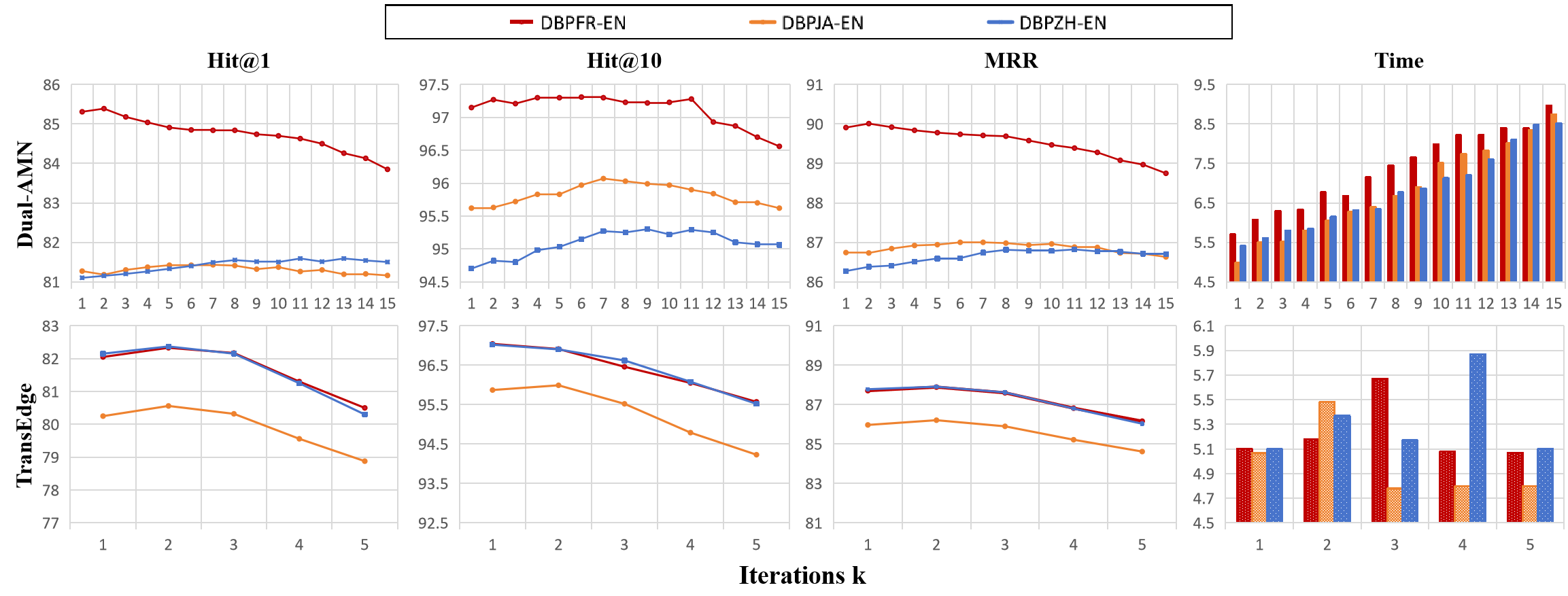}
\caption{Metrics and execution time (seconds) on DBP15K across iterations $k$.} 
\label{iteration}
\vspace{-10pt}
\end{figure*}

In this subsection, we delve into the iterative performance of TFP when using Dual-AMN and TransEdge encoders on the DBP15K. The results, depicting TFP's behavior across different iteration counts $k$, are illustrated in Figure \ref{iteration}.

\vspace{-12pt}
\subsubsection{\textbf{Performance with Dual-AMN Encoder.} }
Dual-AMN stands out in the DBP15K dataset for its performance in Hit@k and MRR metrics. TFP's iteration-dependent performance shows a typical increase followed by a plateau, and sometimes a decrease due to over-smoothing. For example, TFP peaks at $k=11$ iterations on DBP${ZH-EN}$ before stabilizing. However, Hit@10 initially decreases, then stabilizes, showing convergence despite random fluctuations. TFP shows quick convergence and notable improvements on DBP${FR-EN}$, benefiting from its diverse topological structures \cite{sun2017cross}. Time analysis indicates that while computational time increases with iterations, it remains under 10 seconds even at $k=15$, which is minor compared to the encoder's training time.

\vspace{-12pt}
\subsubsection{\textbf{Analysis with TransEdge Encoder.}}
TransEdge, a leading translation-based encoder, requires fewer iterations for TFP to converge and exhibit significant enhancements. Unlike Dual-AMN, TransEdge excels in capturing higher-level semantic information but falls short in detailing topological structures, a common trait among translation-based encoders. TFP effectively compensates for this limitation by reconstructing entity features that assimilate topological information from three distinct perspectives, thereby enriching the quality of entity features. This implies that during the encoding stage, structural noise due to model bias is minimal, reducing the need for extensive denoising during decoding. Interestingly, an inverse proportionality emerges between iteration count and time cost beyond $k=2$, characterized by a diminishing time requirement with increasing iterations. Our comprehensive experimental analysis attributes this phenomenon to TFP's enhancement of homophily in relational features during the decoding phase leading to over-smoothing, particularly at higher iteration counts. This augmentation results in sparser features as per equations (\ref{relationequation}) and (\ref{tripleequation}), culminating in over-smoothing that, in turn, leads to a reduction in computational time due to the increased feature sparsity.

\subsection{Time Complexity Analysis (Q3)}

\begin{table}[t]
\setlength{\tabcolsep}{3pt}
\centering
\scriptsize
\captionof{table}{\small Execution time (s) of EA methods decoding with TFP, where TFP(C) / (G) denotes CPU / GPU execution.}
\label{Timeresult}
\resizebox{\textwidth}{!}{
\begin{tabular}{c|c|c
>{\columncolor[HTML]{EFEFEF}}c 
>{\columncolor[HTML]{EFEFEF}}c| c
>{\columncolor[HTML]{EFEFEF}}c 
>{\columncolor[HTML]{EFEFEF}}c| c
>{\columncolor[HTML]{EFEFEF}}c 
>{\columncolor[HTML]{EFEFEF}}c }
\toprule
& Time/s & AlignE & TFP(C) & TFP(G) & RSN  & TFP(C) & TFP(G) & TransEdge & TFP(C) & TFP(G) \\
\cmidrule{2-11}
& DBP15K & 2087   & 13.9   & 4.8    & 3659 & 14.2   & 4.8    & 1625      & 12.9   & 4.7    \\
\multirow{-3}{*}{\begin{tabular}[c]{@{}c@{}}Translation\\ based\end{tabular}} & SRPRS  & 1190   & 8.1    & 4.2    & 1279 & 9.2    & 3.7    & 907       & 8.7    & 3.7    \\ \midrule
& Time/s & MRAEA  & TFP(C) & TFP(G) & RREA & TFP(C) & TFP(G) & DualAMN   & TFP(C) & TFP(G) \\
\cmidrule{2-11}
& DBP15K & 1743   & 16.6   & 5.9    & 323  & 16.3   & 5.7    & 177       & 17.7   & 5.1    \\
\multirow{-3}{*}{\begin{tabular}[c]{@{}c@{}}GNN\\ based\end{tabular}}         & SRPRS  & 558    & 10.6   & 4.6    & 276  & 11     & 3.8    & 163       & 11.4   & 4.6    \\
\toprule
\end{tabular}}
\end{table}

TFP operates through iterative sparse-to-dense matrix multiplications, positioned as the subsequent step after encoders. Adapting from \cite{mao2022lightea}, we streamline these multiplications in Eqs. (\ref{relationequation}) and (\ref{entityequation}) into the \textit{sparse-dense-multiplications} form to sustain a computational complexity of $O(k(|\mathcal{T}|d_r + |\mathcal{E}|d_e))$. For Eqs. (\ref{tripleequation}) and (\ref{fianlequation}), Tensorflow's sparse matrix multiplications are utilized, bringing the complexity down to approximately $O(|\mathcal{T}|d_r)$ and $O(|\mathcal{E}|d_e)$ respectively.

In Table \ref{Timeresult}, we present a detailed comparison of the time costs associated with the training and decoding phases (TFP) across six EA encoders, utilizing both CPU and GPU, on DBP15K and SRPRS datasets. Notably, the propagation step of TFP constitutes only a minor portion of the overall runtime, with the bulk of the time being allocated to the encoder training process. Remarkably, on a GPU, TFP's execution time peaks at just 5.9 seconds. Even when operating on a CPU, TFP requires a maximum of merely 17.7 seconds. This duration is insignificant when contrasted with even the fastest encoder, Dual-AMN.

A comparative observation reveals that TFP's application is expedited on translation-based encoders relative to GNN-based ones. This acceleration is attributable to the varying entity embedding dimensions discussed in Section \ref{Implementaion Details}. The output dimension for GNN-based encoders, particularly for MRAEA and Dual-AMN, is set at $d=600$ and $d=768$, respectively, which is larger than that of the translation-based encoders. This discrepancy in dimensionality directly influences the time efficiency of TFP, as demonstrated by the faster performance on translation-based models.

\subsection{Additional Information (Q4)}
\label{Q4}

\begin{table}[t]
\vspace{-8pt}
\caption{\small Performances of textual EA methods, which make alignment by additional entity name. The results of baselines are collected from the original papers.}
\label{additionalresults}
\setlength{\tabcolsep}{3pt}
\centering
\resizebox{\textwidth}{!}{
\begin{tabular}{c|ccccccccccccccc}
\toprule
\multicolumn{1}{c|}{Datasets}     & \multicolumn{3}{c}{DBP$_{FR-EN}$}                & \multicolumn{3}{c}{DBP$_{JA-EN}$}                  & \multicolumn{3}{c}{DBP$_{ZH-EN}$}                  & \multicolumn{3}{c}{SRPRS$_{FR-EN}$}               & \multicolumn{3}{c}{SRPRS$_{DE-EN}$}                \\
 \cmidrule(lr){1-1}\cmidrule(lr){2-4}\cmidrule(lr){5-7}\cmidrule(lr){8-10}\cmidrule(lr){11-13} \cmidrule(lr){14-16}
Models       & H@1      & MRR      & Time/s    & H@1      & MRR      & Time/s    & H@1      & MRR      & Time/s    & H@1       & MRR       & Time/s    & H@1       & MRR       & Time/s \\
\midrule
GM-Align \cite{xu2019cross}     & 89.4     & -        & 26328     & 73.9     & -        & 26328     & 67.9     & -        & 26328     & 57.4      & 60.2      & 13032     & 68.1      & 71.0      & 13032     \\
RDGCN \cite{wu2019relation}        & 87.3     & 90.1     & 6711      & 76.3     & 76.3     & 6711      & 69.7     & 75.0     & 6711      & 67.2      & 71.0      & 886       & 77.9      & 82.0      & 886       \\
HGCN \cite{wu2019jointly}         & 76.6     & 81.0     & 11275     & 89.2     & 91.0     & 11275     & 72.0     & 76.0     & 11275     & 67.0      & 71.0      & 2504      & 76.3      & 80.1      & 2504      \\
AtrrGNN \cite{liu2020exploring}      & 91.9     & 91.0     & -         & 78.3     & 83.4     & -         & 79.6     & 84.5     & -         & -         & -         & -         & -         & -         & -         \\
EPEA \cite{wang2020knowledge}         & 95.5     & 96.7     & -         & 92.4     & 94.2     & -         & 88.5     & 91.1     & -         & -         & -         & -         & -         & -         & -         \\
SEU \cite{mao2021alignment}        & 98.8     & 99.2     & 17.0      & 95.6     & 96.9     & 16.2      & 90.0     & 92.4     & 16.2      & 98.2      & 98.6      & 9.6       & 98.3      & 98.7      & 9.8       \\
LightEA \cite{mao2022lightea}      & \textbf{99.5}     & \underline{99.6}     & 15.7      & \textbf{98.1}     & \textbf{98.7}     & 14.8      & \textbf{95.2}     & \textbf{96.4}     & 15.2      & \underline{98.6 }     & 98.9      & 11.2     & \underline{98.8 }     & \underline{99.1 }     & 11.4     \\
\midrule
\rowcolor[HTML]{EFEFEF} 
TFP(SEU)     & 99.05       & 99.4  & 4.9    & 96.21       & 97.19 & 4.7    & 90.51       & 92.63 & 5.1    & 98.55         & 98.95 & 3.7    & 98.57         & 98.95 & 3.8    \\
\rowcolor[HTML]{EFEFEF} 
TFP(LightEA) & 99.21       & 99.65 & 5.1    & \underline{97.94 }      & \textbf{98.7}  & 4.9    & \underline{94.34}       & \underline{95.62} & 5.4    & \underline{98.6}          & \underline{99.03} & 3.9    & 98.6          & 99    & 3.8 \\
\rowcolor[HTML]{EFEFEF} 
TFP & \underline{99.33 }    & \textbf{99.68}   & \textbf{4.8}       & 96.68     & \underline{97.54}    & \textbf{4.7}       & 91.54    & 93.59    & \textbf{4.9}       & \textbf{98.99}      & \textbf{99.28}      & \textbf{3.6}       & \textbf{98.87}      & \textbf{99.15}      & \textbf{3.6}       \\
\bottomrule
\end{tabular}}
\end{table}

Within the scope of EA, most experiments have largely focused on purely structural-based methods. Nonetheless, some studies \cite{xu2019cross, wu2019jointly} have integrated additional information, such as entity names, to boost performance. These approaches often utilize machine translation systems or cross-lingual word embeddings to convert entity names into a unified semantic space, using averaged pre-trained word embeddings to establish initial features for entities and relations. In this context, the initial entity features are pre-trained, making these textual EA methods resemble decoding algorithms that focus on topology and dependencies. TFP is capable of fulfilling a similar role, akin to scenarios where seed alignment sets are absent.

To ensure fair comparisons with textual-based EA methods, we used pre-trained GLoVe word embeddings \cite{pennington2014glove} for entity names as the initial entity features $\mathbf{X}_e^{(0)}$ decoding by TFP, treating it as an unsupervised textual-based method since no seed alignments are required. TFP variants such as TFP (SEU) and TFP (LightEA) use SEU and LightEA as encoders, respectively.

Table \ref{additionalresults} shows TFP's performance relative to seven other baseline textual-based methods across the DBP15K and SRPRS datasets. As an unsupervised textual EA method, TFP outperforms all competitors on SRPRS with minimal time cost and ranks near the top on DBP15K. The comparative performance of SEU, LightEA, and TFP highlights their formidable competitiveness, suggesting that propagation strategies may be more effective than traditional neural network approaches, questioning the necessity of complex neural networks and seed alignments when using textual features. Notably, TFP achieves these superior results in under 5 seconds, setting a new standard for SOTA performance on SRPRS—a dataset representative of sparse KGs commonly encountered in real-world applications.

However, TFP(SEU) and TFP(LightEA) show a slight decline in performance when compared to using GLoVe embeddings alone. This can be attributed to the inherent similarities between SEU, LightEA, and TFP: all of them operate on propagation principles, ostensibly acting as smoothing strategies predicated on topological data. In scenarios involving sparse KGs, SEU's and LightEA`s ability to comprehensively capture topological nuances and facilitate entity embedding smoothing is significant. However, subsequent application of TFP may precipitate over-smoothing, which elucidates the observed performance discrepancy.

\section{Conclusion}
\label{conclusion}
In this paper, we presented Triple Feature Propagation (TFP), an innovative, effective and theoretical approach for entity alignment decoding. TFP extends traditional adjacency matrices into multi-view matrices—encompassing entity-to-entity, entity-to-relation, relation-to-entity, and relation-to-triples relationships—to reconstruct entity embeddings. This reconstruction, aimed at homophily maximization, is achieved by minimizing the Dirichlet energy, facilitating a natural feature propagation through gradient flow theory. Our experiments confirm TFP's capability to improve the performance of various EA methods, with less than 6 seconds of additional computational time. This efficiency and general applicability make TFP a promising advancement in the field of entity alignment.

% \noindent
% \textit{Supplemental Material Statement:} 
\paragraph*{Supplemental Material Statement:}
The source code and constructed datasets are available at our repository:
https://github.com/wyy-code/TFP
% https://anonymous.4open.science/r/TFP-9425 
The proofs of our proposition and other details are provided in the Appendix file.

% %-------------------Acknowledge--------------------
% \begin{credits}
% \subsubsection{\ackname} A bold run-in heading in small font size at the end of the paper is
% used for general acknowledgments, for example: This study was funded
% by X (grant number Y).

% \subsubsection{\discintname}
% It is now necessary to declare any competing interests or to specifically
% state that the authors have no competing interests. Please place the
% statement with a bold run-in heading in small font size beneath the
% (optional) acknowledgments\footnote{If EquinOCS, our proceedings submission
% system, is used, then the disclaimer can be provided directly in the system.},
% for example: The authors have no competing interests to declare that are
% relevant to the content of this article. Or: Author A has received research
% grants from Company W. Author B has received a speaker honorarium from
% Company X and owns stock in Company Y. Author C is a member of committee Z.
% \end{credits}

\newpage
\bibliographystyle{splncs04}
\bibliography{reference.bib}

\newpage
%% If your work has an appendix, this is the place to put it.
\appendix

\section{Proof of Proposition 1}
\textit{Proposition 1.} \textbf{(Existence of the solution.)} The matrix $\Delta_{oo}$ is non-singular, allowing the reconstruction of other entity features $\mathbf{x}_o$ using seed alignment entity features $\mathbf{x}_s$ as ${\mathbf{x}}_{o}(t) = -\mathbf{\Delta}^{-1}_{oo}\mathbf{\Delta}_{os}{\mathbf{x}}_s$.
\begin{proof}
Initially, consider an undirected connected graph scenario where $\mathbf{\Delta}_{oo}$ is a sub-matrix of the Laplacian matrix $\mathbf{\Delta}$. Given that sub-Laplacian matrices of undirected connected graphs are invertible  \cite{rossi2022unreasonable}, thus $\Delta_{oo}$ is non-singular. The spectral properties of eigenvalues in undirected graphs suggest similar non-singularity for directed graphs \cite{maskey2023fractional}.

For general cases, assume an ordered representation of the adjacency matrix for a disconnected graph as:
\begin{equation}
\label{disconnected}
\mathbf{A} = \text{diag}(\mathbf{A}_1,\dots,\mathbf{A}_r)
\end{equation}
Here, $\mathbf{A}_{i}, i={1,\dots,r}$, represents each connected component. The gradient flow in equation (\ref{solution}) is applicable to each connected component independently for disconnected graphs.
\end{proof}

\section{Proof of Proposition 2}
\textit{Proposition 2.} 
\textbf{(Approximation of the solution.)} 
With the iterative reconstruction solution as delineated in equation (\ref{itersolution}), and considering a sufficiently large iteration count $N$, the entity features will approximate the results of feature propagation as follows:
\begin{equation}
    \mathbf{X}^{(N)} \approx 
    \begin{pmatrix}
        \mathbf{x}_{s} \\ \mathbf{\Delta}^{-1}_{oo}\mathbf{\widetilde{A}}_{os}{\mathbf{x}}_s\\
    \end{pmatrix}
\end{equation}
\begin{proof}
Commencing with the initial entity features $\mathbf{X}^{(0)}$ generated by EA encoders and applying equation (\ref{itersolution}), we iterate:
\begin{equation}
\begin{aligned}
    \begin{pmatrix}
        \mathbf{x}^{(k)}_s \\
        \mathbf{x}^{(k)}_{o}\\
    \end{pmatrix}
     & = 
     \begin{pmatrix}
         \mathbf{I} & \mathbf{0}\\
        \mathbf{\widetilde{A}}_{os} & \mathbf{\widetilde{A}}_{oo}\\
     \end{pmatrix} 
    \begin{pmatrix}
        \mathbf{x}^{(k-1)}_s \\
        \mathbf{x}^{(k-1)}_{o}\\
    \end{pmatrix}
    = \begin{pmatrix}
    \mathbf{x}^{(k-1)}_s \\ \mathbf{\widetilde{A}}_{os}\mathbf{x}^{(k-1)}_s + \mathbf{\widetilde{A}}_{oo}\mathbf{x}^{(k-1)}_o \\
\end{pmatrix}
\end{aligned}
\end{equation}
Given the stationary nature of the seed alignment entity features $\mathbf{x}_s$, we have the equation $\mathbf{x}^{(k)}_s = \mathbf{x}^{(k-1)}_s=\mathbf{x}_s$. The focus then shifts to the convergence of $\mathbf{x}_o$::
\begin{equation}
    \mathbf{x}^{(k)}_o = \mathbf{\widetilde{A}}_{os}\mathbf{x}_s + \mathbf{\widetilde{A}}_{oo}\mathbf{x}^{(k-1)}_o
\end{equation}
Expanding and analyzing the limit for the stationary state, we find:
\begin{equation}
\begin{aligned}
    \lim\limits_{k \to \infty} \mathbf{x}_o^{(k)} & = \mathbf{\widetilde{A}}_{os}\mathbf{x}_s + \lim\limits_{k \to \infty}\sum_{i=2}^k\mathbf{\widetilde{A}}_{oo}^{i-1}\mathbf{\widetilde{A}}_{os}\mathbf{x}_s + \lim\limits_{k \to \infty}\mathbf{\widetilde{A}}_{oo}^k\mathbf{x}^{(0)}_o \\
    & = \lim\limits_{k \to \infty}\mathbf{\widetilde{A}}_{oo}^k\mathbf{x}^{(0)}_o + \lim\limits_{k \to \infty}\sum_{i=1}^k\mathbf{\widetilde{A}}_{oo}^{i-1}\mathbf{\widetilde{A}}_{os}\mathbf{x}_s
\end{aligned}
\end{equation}
Spectral graph theory provides critical insights into the properties of the Laplacian matrix $\mathbf{\Delta}$. It establishes that the eigenvalues of $\mathbf{\Delta}$ are confined within the range [0,2). This spectral characteristic has direct implications for the matrix $\mathbf{\widetilde{A}} = \mathbf{I} - \mathbf{\Delta}$, whose eigenvalues are consequently within the range (-1,1]. A pivotal aspect of this discussion, as elucidated in Proposition \ref{existence}, is the non-singularity of $\mathbf{\Delta}_{oo}$. The absence of 0 as an eigenvalue of $\mathbf{\Delta}_{oo}$ implies that $\mathbf{\widetilde{A}}$'s eigenvalues strictly occupy the interval (-1,1), thereby excluding the endpoints. This spectral behavior significantly influences the convergence properties of the iterative process. Specifically, the limit $\lim\limits_{k \to \infty}\mathbf{\widetilde{A}}_{oo}^n\mathbf{x}^{(0)}_o$ approaches 0. Furthermore, the summation $\lim\limits_{k \to \infty}\sum_{i=1}^k\mathbf{\widetilde{A}}_{oo}^{i-1}$ converges to $(\mathbf{I}-\mathbf{\widetilde{A}_{oo}})^{-1} = \mathbf{\Delta}^{-1}_{oo}$. By integrating these insights, the long-term behavior of the iterative solution can be articulated as:
\begin{equation}
    \lim\limits_{k \to \infty} \mathbf{x}_o^{(k)} = \mathbf{\Delta}^{-1}_{oo}\mathbf{\widetilde{A}}_{os}{\mathbf{x}}_s
\end{equation}
Therefore, when the number of iterations $N$ is sufficiently large, the entity features in $\mathbf{x}_o^{(N)}$ approximate $\mathbf{\Delta}^{-1}_{oo}\mathbf{\widetilde{A}}_{os}{\mathbf{x}}_s$.
\end{proof}

\section{Alignment Search}
\label{Alignment Search}
In the testing process, rather than use the popular distance metric of Cross-domain Similarity Local Scaling (CSLS)\cite{lample2018word} to search the alignments in most works, we follow \cite{mao2022effective} and \cite{mao2022lightea} to formalize the entity alignment problem as an assignment problem to enforce the one-to-one alignment constraint. Before giving the mathematical definition, it assumes $|\mathcal{T}_s| = |\mathcal{T}_t|=n^{t}$ to simplify the process. In addition, it uses $SIM\in \mathbb{R}^{n^t\times n^t}$ to represent the cosine similarity matrix and compute it between testing entities in two $\
\mathcal{KG}$s with the entity embeddings.  Thus, we attempt to solve the following optimization problem:
  \begin{equation}
 \text{arg}\,\max\limits_{P\in \mathbb{P}_{n^t}}\left \langle P,SIM \right \rangle
 \end{equation}
 where $\mathbb{P}_{n^t}$ is a set of permutation matrices with shape of $\mathbb{R}^{n^t \times n^t}$. Actually, we can directly obtain the optimal solution $P^{\star}$ by Sinkhorn operation\cite{cuturi2013sinkhorn}.
   \begin{equation}
P^{\star} = \lim \limits_{\tau \rightarrow 0^{+}} \text{Sinkhorn}(\frac{SIM}{\tau})
 \end{equation}
 where the operation of Sinkhorn is as follows:
 \begin{equation}
     Sinkhorn(\mathbf{X}) =  \lim \limits_{k \rightarrow +\infty} S^k(\mathbf{X}), S^k(\mathbf{X})=\mathcal{N}_c(\mathcal{N}_r(S^{k-1}(\mathbf{X})))
 \end{equation}
 where $S^0(\mathbf{X}) = exp(\mathbf{X})$, $\mathcal{N}_r(\mathbf{X})=X\oslash (\mathbf{X}\mathbf{1}_{n^{t}}\mathbf{1}_{n^{t}}^{\top})$ and $\mathcal{N}_c(\mathbf{X})=X\oslash (\mathbf{1}_{n^{t}}\mathbf{1}_{n^{t}}^{\top}\mathbf{X})$ are the row and column-wise normalization operators of a matrix, $\oslash$ denotes the element-wise division, $\mathbf{1}_{n^{t}}$ is a column vector of ones. Though we can only obtain an approximate solution with a small $k$ in practice, we empirically found that the approximate solution is enough to obtain a good alignment performance.
 Considering that the assumption of $n^t=|\mathcal{T}_s|=|\mathcal{T}_t|$ is easily violated, a naive reduction is to pad the similarity matrix with zeros such that its shape becomes $\mathbb{R}^{n^t\times n^t}$ where $n^t=max(n^t_1,n^t_2)$.

\section{Evaluation Metric}
\label{Evaluation Metric}

We utilize cosine similarity to calculate the similarity between two entities and employ H@k and MRR as metrics to evaluate all the methods. H@k describes the fraction of truly aligned target entities that appear in the first $k$ entities of the sorted rank list:
\begin{equation}
    H@k = \frac{1}{|S_t|}\sum_{i=1}^{|S_t|}\mathbb{I}[rank_{i}\le k]
\end{equation}
where $rank_i$ refers to the rank position of the first correct mapping for the $i$-th query entities and $\mathbb{I} = 1$ if $rank_{i}\le k$ and $0$ otherwise. $S_t$ refers to the testing alignment set.
MRR (Mean Reciprocal Ranking) is a statistical measure for evaluating many algorithms that produce a list of possible responses to a sample of queries, ordered by the probability of correctness. In the field of EA, the reciprocal rank of a query entity (i.e., an entity from the source KG) response is the multiplicative inverse of the rank of the first correct alignment entity in the target KG. MRR is the average of the reciprocal ranks of results for a sample of candidate alignment entities:
\begin{equation}
    MRR = \frac{1}{|S_t|}\sum_{i=1}^{|S_t|}\frac{1}{rank_i}
\end{equation}

\end{document}